\documentclass[notitlepage,floatfix,aps,pra,reprint,superscriptaddress,twocolumn,footinbib,longbibliography]{revtex4-2}
\usepackage[utf8]{inputenc}
\usepackage{amsmath,amsfonts,amssymb,amscd,mathtools}
\usepackage{amsthm}
\usepackage{dsfont}
\newtheorem{theorem}{Theorem}

\newtheorem{lemma}{Lemma}

\newtheorem{definition}{Definition}
\usepackage{enumerate}
\usepackage{braket}
\usepackage{bm}
\usepackage{graphicx}
\usepackage{verbatim}
\usepackage{url}
\usepackage{hyperref}
\usepackage[dvipsnames]{xcolor}
\hypersetup{
    pdfnewwindow=true,
    colorlinks=true,
    citecolor=Blue,
    linkcolor=Blue,
    urlcolor=Blue
}
\usepackage{tikz}
\DeclareMathOperator{\Tr}{Tr}
\DeclareMathAlphabet{\mathpzc}{OT1}{pzc}{m}{it}
\DeclareMathAlphabet{\mathcalligra}{T1}{calligra}{m}{n}

\newcommand{\cF}{\mathcal{F}}
\newcommand{\cH}{\mathcal{H}}

\newcommand{\cP}{\mathcal{P}}
\newcommand{\cS}{\mathcal{S}}
\newcommand{\cT}{\mathcal{T}}

\newcommand{\rA}{\mathrm{A}}
\newcommand{\rB}{\mathrm{B}}

\begin{document}

\title{Von Neumann's information engine without the spectral theorem}
\author{Shintaro Minagawa}
\email{minagawa.shintaro@nagoya-u.jp}
\affiliation{Graduate School of Informatics, Nagoya University, Chikusa-Ku, Nagoya 464-8601, Japan}
\author{Hayato Arai}
\email{m18003b@math.nagoya-u.ac.jp}
\affiliation{Graduate School of Mathematics, Nagoya University, Chikusa-Ku, Nagoya 464-8602, Japan}
\author{Francesco Buscemi}
\email{buscemi@nagoya-u.jp}
\affiliation{Graduate School of Informatics, Nagoya University, Chikusa-Ku, Nagoya 464-8601, Japan}

\begin{abstract}
Von Neumann obtained the formula for the entropy of a quantum state by assuming the validity of the second law of thermodynamics in a thought experiment involving semipermeable membranes and an ideal gas of quantum-labeled particles. Despite being operational in the most part, von Neumann's argument departs from an operational narrative in its use of the spectral theorem. In this work we show that the role of the spectral theorem in von Neumann's argument can be taken over by the operational assumptions of repeatability and reversibility, and using these we are able to explore the consequences of the second law also in theories that do not possess a unique spectral decomposition. As a byproduct, we obtain the Groenewold--Ozawa information gain as a natural monotone for a suitable ordering of instruments, providing it with an operational interpretation valid in quantum theory and beyond.
\end{abstract}

\maketitle

\section{Introduction}
Ever since Maxwell summoned his demon~\cite{maxwell1871theory}, the mutual influence between physics (i.e., the representation of a system's physical properties), information (i.e., the representation of an agent's knowledge about a physical system), and the measurement process (i.e., the interaction between system and agent) has emerged as one of the main themes of debate in theoretical physics. Von Neumann is surely among the most influential names to have contributed to this discussion. In his mathematical formulation of quantum theory~\cite{von1955mathematical}, much space is devoted to a careful analysis of the interplay between quantum measurement theory and thermodynamics. Therein, von Neumann approaches the problem using the artifact---common, as he explicitly remarks, in phenomenological thermodynamics and used before him also by Einstein~\cite{einstein1914beitrage} (see also Ref.~\cite{klein1967thermodynamics}) and Szilard~\cite{szilard1929uber,szilard1964decrease}---of an ideal gas of particles, whose mechanical degrees of freedom obey the laws of classical mechanics, while their states (which should really be thought of, in this context, as mere labels) are described according to quantum theory, but are otherwise irrelevant from an energetic viewpoint. In this way, it is possible to separate, on the one side, the mechanical and thermal properties of the gas, and, on the other side, the information that an agent acting on the gas has about its particles. The ``missing link'' between physics and information is provided by two assumptions: the first is the existence of particular devices called \textit{semipermeable membranes}; we will discuss them extensively in what follows. The second assumption is about the validity of the second law of thermodynamics, which is posited by von Neumann \textit{ab initio}. Following this line of thought, von Neumann was able to explore the consequences of the second law in quantum theory and obtain his famous formula for the entropy of quantum states.

The argument constructed by von Neumann, although operational in the most part (i.e., the quantum entropy is defined using a thermodynamic protocol, which in principle also provides a way to measure the quantum entropy), still relies on the structure of Hilbert spaces. In particular, a crucial role is played by the spectral decomposition of self-adjoint operators~\cite{von1955mathematical}.
A natural question is then to see how far von Neumann's discussion can be reconstructed from purely operational assumptions. The common way to approach this kind of problems utilizes the framework of general probabilistic theories (GPTs; see, e.g., Ref.~\cite{dariano-OPT,janotta2014generalized,plavala2021general,popescu1994quantum,pawlowski2009information,barnum2010entropy,masanes2011derivation,muller2012structure,barnum2014higher,wakakuwa2021gentle}). These provide a modern take on the operational approach to quantum theory, which can be traced back to works by Ludwig~\cite{ludwig1964,ludwig1967}, Davies and Lewis~\cite{davies1970operational}, Gudder~\cite{gudder1973}, and Ozawa~\cite{ozawa1980optimal}.

Without Hilbert spaces, in GPTs there exist various ways to introduce the entropy functional in an operational way, which turn out to be all equivalent in conventional quantum theory, but are not so in general \cite{short2010entropy,barnum2010entropy,kimura2010distinguishability,kimura2016entropies}.
A possible approach is to add assumptions that are strong enough to conclude that a unique spectral decomposition exists~\cite{chiribella2015operational,krumm2017thermodynamics}.
However, in more general setups, the uniqueness of the spectral decomposition is not guaranteed~\cite{barnum2014higher,barnum2015entropy,krumm2015thermodynamics}: in all such cases, von Neumann's argument seems to be a nonstarter.

In this paper, we clarify the importance of two implicit assumptions in von Neumann's argument: the existence of repeatable measurement processes, on the one hand, and of states that are the fixed point of some non-trivial repeatable measurement, on the other. These two operational assumptions---we argue---can take over the role played, in von Neumann's argument, by the spectral theorem, which instead is not operational. In this way, we can provide von Neumann's thought experiment with a fully operational narrative, and to explore the consequences of the second law of thermodynamics also in GPTs that do \textit{not} have a unique spectral decomposition. However, to achieve this, some modifications to von Neumann's argument are needed: in particular, the thermodynamic process must be modified into a cycle. As a byproduct, our argument allows us to obtain also an analogue of the Groenewold--Ozawa information gain~\cite{groenewold1971problem,ozawa1986information} in a wide range of GPTs, and to equip it with the operational meaning of monotone, with respect to a suitable preorder of the measurement processes.

\section{Basic definitions}
In order to discuss von Neumann's thought experiment in GPTs, we begin by briefly reviewing the basics of a simple single-system theory (see, e.g., Refs.~\cite{janotta2014generalized,plavala2021general} and references therein). A single-system GPT is determined by providing all possible \textit{events}, thought of as ``black boxes'' with an input and an output, which can be either the system of the theory or the trivial system (i.e., a system with only one possible state). The input-output arrangement determines the \textit{type} of the event. Events with trivial input are interpreted as preparation events, whereas events with trivial output are interpreted as observation events, or \textit{effects}. Families of events of the same type form a \textit{test}: one has, therefore, preparation tests and observation tests. The latter are usually called \textit{measurements}. Tests describe what can happen in an experiment: among the events it contains, one and only one will occur in any given repetition of the same experiment. Therefore, tests with only one element describe \textit{deterministic} events. In particular, a common assumption (corresponding to a ``no-signaling from the future'' principle~\cite{no-signaling-future,dariano_2010}) is that only one deterministic measurement exists. Instead, a theory typically provides many possible deterministic preparations: these are the (normalized) \textit{states} of the theory and form its \textit{state space}, denoted by $\Omega$. Normalized states naturally form a convex set: its extremal points are called \textit{pure states}, otherwise they are \textit{mixed}.

Tests can be composed, following the idea that any experiment can be seen as a chain preparation--process--measurement. Since states can be convexly mixed, tests are naturally assumed to satisfy linearity on such convex combinations. This observation allows for a more concrete definition of single-system tests as families $\{s_j\}_{j\in J}$ of affine maps acting on $\Omega$~\footnote{Notice that we are not assuming that \textit{any} family of affine maps constitutes a legitimate test of the theory: further conditions, most notably that of complete positivity, may restrict the set of possible tests. However, this point is irrelevant for the present discussion.}. Denoting by $u$ the unique deterministic measurement, normalization of probability requires that $\sum_j(u\circ s_j)\rho=1$, for all $\rho\in\Omega$. In conventional quantum theory, it is easy to recognize that tests are \textit{quantum instruments}~\cite{ozawa1984quantum}, whereas the composition $\{u\circ s_j\}_{j\in J}$ provides the generalization of \textit{positive operator-valued measures} (POVMs). For this reason, in what follows we will call \textit{instruments} those tests that have both input and output non-trivial, while the term \textit{measurement} will only be used to denote the analogue of POVMs.
Summarizing, in what follows we will work with normalized states, denoted by $\rho$, $\sigma$ etc., instruments, denoted by $\{s_j\}_{j\in J}$, $\{t_k\}_{k\in K}$ etc., and measurements, denoted by $\{e_j\}_{j\in J}$, $\{f_k\}_{k\in K}$ etc.

Finally, an important notion is that of perfect distinguishability: a family of normalized states $\{\rho_j\}_{j\in J}$, $\rho_j\in\Omega$ is said to be \emph{perfectly distinguishable} if there exists a measurement $\{e_j\}_{j\in J}$ such that \[e_{j}(\rho_{j'})=\delta_{jj'}\;,\] for all $j,j'\in J$.

\medskip\textit{The measurement stage.}---As anticipated in the introduction, von Neumann's thermodynamic thought experiment makes use of semipermeable membranes (SPMs). These are devices that are can separate, reversibly and without any thermodynamic cost, the particles of a gas, as long as their states are distinguishable---at least in principle, of course. Von Neumann goes to a great length to justify the use of such idealized devices, which, he argues, represent ``\textit{the thermodynamic definition of difference}''~\cite{von1955mathematical}. In what follows, we characterize SPMs from an operational viewpoint.

The first property of SPMs, implicit in their definition, is their \textit{repeatability}: the first time a particle collides with the membrane, a measurement occurs determining whether the arriving particle is a ``pass'' or a ``bounce,'' and it will remain so until the end of the experiment, even if the same particle collides with the membrane multiple times. Hence, our first assumption is that the theory contains repeatable instruments. Formally, an instrument $\{s_j\}_{j\in J}$ is repeatable if and only if \[s_{j'}\circ s_{j}=\delta_{j'j} s_j\;,\] holds, for all $j,j'\in J$.

The second requirement needed to go along with von Neumann's discussion is the existence of states that are the fixed points of some repeatable instrument.  This is necessary if we want to speak, as von Neumann does, of thermodynamic reversibility with respect to the system's initial state. Since even in quantum theory there exist repeatable instruments without invariant states~\cite{buscemi2004repeatable}, we need to treat this requirement as a further assumption. In thermodynamic terms, the assumption of state invariance makes it reasonable to assume that there are no hidden thermodynamic costs (besides the macroscopic mechanical ones) incurred when using SPMs that preserve the mixture being separated. Formally, we say that an instrument $\{s_j\}_{j\in J}$ is $\rho$-\textit{preserving} whenever $\sum_js_j(\rho)=\rho$.

In the assumption of state invariance, however, we need to exclude the trivial case of the deterministic identical instrument, which is of course repeatable and for which all states are fixed points. We thus strengthen our requirements as follows: first, we define a set of ``maximal'' instruments with respect to a suitable preorder, and then require the existence of states that are fixed points of some instrument that is maximal \textit{and} repeatable. As a straightforward extension of the post-processing preorder of POVMs~\cite{buscemi2005clean} to the case of instruments, we introduce the following preorder (see also Fig.~\ref{fig:refinement}):

\begin{figure}[htbp]
    \centering
    \includegraphics[width=7cm]{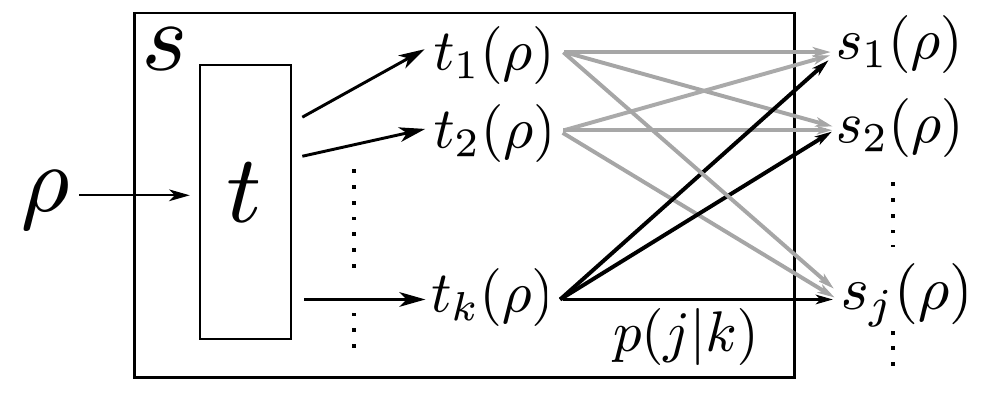}
    \caption{A schematic representation of the notion of Groenewold majorization $t\succ_\rho s$.}
    \label{fig:refinement}
\end{figure}

\begin{definition}[Groenewold majorization]
	\label{def:rho-ref}
	Given two instruments $s=\{s_j\}_{j\in J}$ and $t=\{t_k\}_{k\in K}$, and a state $\rho\in\Omega$, we say that $t$ \emph{Groenewold-majorizes} $s$ given $\rho$, in formula,
	\[t\succ_\rho s\;,\] if and only if there exists a conditional probability distribution $p(j|k)$ such that
	\begin{equation}
		s_j(\rho)=\sum_{k\in K}p(j|k) t_k(\rho)\;,
	\end{equation}
for all $j\in J$.
\end{definition}

In the above definition, we could have in fact chosen a more general scenario, including also some suitable transformations before and after the randomization. However, since for our purposes we are only interested in the maximal points of $\succ_\rho$ and since these are the same with or without the extra transformations, for the sake of simplicity, we choose to work with the above definition.

\begin{definition}[Fine-grained instruments]
	An instrument $s=\{s_j\}_{j\in J}$ is said to be \emph{fine-grained} if and only if, for all preorders $\succ_\rho$ (i.e., for all states $\rho$), the condition $t\succ_\rho s$ implies that also $s\succ_\rho t$ holds.
\end{definition}

\begin{definition}[MPP instruments]
	An instrument $s=\{s_j\}_{j\in J}$ is said to be of the \emph{measure-and-prepare-pure (MPP)} form if and only if it is completely characterized by one measurement $\{e_j\}_{j\in J}$ and one family of normalized pure states $\{\sigma_j\}_{j\in J}$, such that the $j$-th state $\sigma_j$ is prepared whenever the $j$-th effect $e_j$ occurs.
\end{definition}

The following lemma holds (for the proof, see Appendix \ref{app:lamma:f-g}) .

\begin{lemma}\label{lemma:f-g-instrument}
	Let $s=\{s_j\}_{j\in J}$ be a fine-grained instrument. Then, $s_j(\rho)$ is (up to normalization) a pure state, for all states $\rho\in\Omega$ and all $j\in J$. In other words, $s$ is an \emph{MPP} instrument. 
\end{lemma}

The above lemma guarantees that fine-grained instruments are all physically admissible, simply because they can be physically realized as measurements followed by the preparation of pure states that only depend on the outcome. As such, they do not require any notion of composition (i.e., complete positivity) to be discussed. We denote fine-grained instruments as pairs $\{e_j,\sigma_j\}_{j\in J}$, where $\{e_j\}$ is a measurement and $\{\sigma_j\}$ is a family of normalized pure states.

Hence, the following two assumptions, that is,
\begin{enumerate}
	\item the existence of fine-grained and repeatable instruments; and,
	\item the existence of states that are fixed points of some fine-grained and repeatable instrument,
\end{enumerate}
together with Lemma~\ref{lemma:f-g-instrument}, imply that the theory contains states that can be decomposed on a set of perfectly distinguishable pure (PDP) states. Following~\cite{barnum2014higher,krumm2017thermodynamics}, we call such states \textit{weakly spectral}. For simplicity, we summarize the previous discussion into one assumption as follows:
	
\textit{Assumption 1} (Weak spectrality). We assume that the theory contains weakly spectral states $\rho\in\Omega$, that is, states that admit a (possibly non-unique) convex decomposition into PDP states $\rho=\sum_i p_i\rho_i$.
	 
Notice that we are not assuming that \textit{all} states of the theory are weakly spectral. In what follows, for ease of notation, for a weakly spectral $\rho$, we denote by $\mathcal D(\rho)$ the set of all possible probability distributions $\{p_i\}$ that appear in at least one of its PDP decompositions.

Assumption 1 (A1) above is ``weak'' because in the literature its ``strong'' version is often encountered, and the separation between the two is strict. More precisely, instead of A1, the property of \textit{(unique or strong) spectrality} assumes that all PDP decompositions of the same state correspond to distributions $\{p_i\}$ which differ at most in a permutation of the indices \cite{barnum2015entropy,krumm2015thermodynamics}.
Strong spectrality is hence akin to assuming the spectral theorem from the onset. Instead, the existence of weakly spectral states is guaranteed as soon as there exist at least two perfectly distinguishable pure states. Examples of theories that satisfy weak spectrality but not strong spectrality are given in Ref.~\cite{krumm2015thermodynamics}. For what follows, it is convenient to introduce the following definition:
\begin{definition}[$\rho$-separating SPMs]
	A set of SPMs is said $\rho$-\emph{separating} if it corresponds to a repeatable \emph{MPP} instrument $\{e_j,\sigma_j\}_{j\in J}$, which is in particular $\rho$-preserving, that is
	\[
	\rho=\sum_{j\in J}e_j(\rho)\sigma_j\;.
	\]
\end{definition}
\noindent From the above discussion, it is clear that $\rho$-separating SPMs exist if only if $\rho$ is weakly spectral.
\begin{figure*}[t]
\centering
\includegraphics[width=\linewidth]{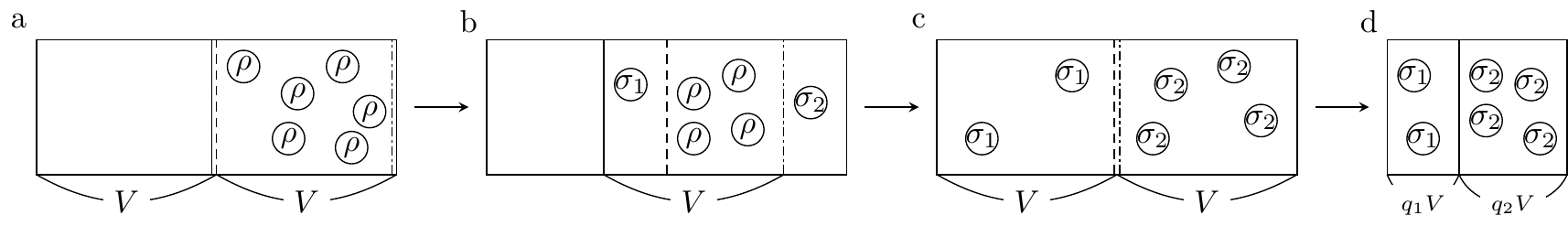}
\caption{The process of separating particles in a state $\rho=q_1\sigma_1+q_2\sigma_2$ using suitable $\rho$-separating SPMs. The dashed line is a SPM transparent to $\sigma_1$ but perfectly elastic for $\sigma_2$; the dashed-dotted line vice versa.}
\label{fig:von_neumann}
\end{figure*}

\section{The feedback stage}
Von Neumann's argument also involves a feedback control stage, during which a suitable transformation is applied to the system, depending on the measurement outcome. Again, to follow von Neumann's narrative, we need an assumption, that we identify in the following:

\textit{Assumption 2} (Free pure-state transformations~\cite{masanes2011derivation,hanggi2013violation,krumm2017thermodynamics}). For any pair of pure states, $\rho_{\textrm{in}}$ and $\rho_{\textrm{out}}$, the theory contains a deterministic event $F:\Omega\to \Omega$ which is reversible, that is, there exists another deterministic event $G$ such that $G\circ F=\operatorname{id}$, and satisfies $F(\rho_\textrm{in})=\rho_{\textrm{out}}$.

According to the thermodynamic narrative, a reversible operation is one that does not cause any change to the entropy of the thermodynamic universe. Notice that we do not put any constraints on what the operation does on states other than $\rho_{\textrm{in}}$. In this sense, Assumption 2 (A2), like A1 before, is ``weak'': instead of A2, in the literature it is common to find the assumption of \textit{strong symmetry}, which assumes that any two collections of PDP states are connected by one \textit{simultaneous} reversible process \cite{barnum2014higher,krumm2015thermodynamics,krumm2017thermodynamics}.
A counterexample of a theory that satisfies A2 but is not strongly symmetric is given in Appendix \ref{sec:app2}, where we explicitly construct a GPT that contains two pairs of PDP states that cannot be simultaneously and reversibly converted.

We conclude this section by noticing that conventional quantum theory satisfies both strong spectrality and strong symmetry. In fact, it is known that any GPT satisfying both strong spectrality and strong symmetry becomes to a large extent akin to quantum theory~\cite{barnum2020spectral}. 

\section{Entropy from thermodynamic considerations}
We are now ready to formulate our version of von Neumann's thought experiment in the language of GPTs. As already noticed, we follow von Neumann's argument, in that the particles' mechanical degrees of freedom obey the classical laws of ideal gases, whereas the non-classical degrees of freedom, i.e., the generalized states labeling the different ``isomers,'' are thought of as internal degrees of freedom with a completely degenerate Hamiltonian so that they do not directly enter in the energetic balance of the process.

We begin with the calculation of the work needed to separate the particles by means of SPMs. The separation process is depicted in Fig.~\ref{fig:von_neumann}.
An ideal, thermostatted gas contains $N$ particles in the mixture state $\rho$. According to the above discussion, $\rho$ is assumed to satisfy the property of weak spectrality. For simplicity, we assume that $\rho$ contains only two pure components, that is, $\rho=q_1\sigma_1+q_2\sigma_2$, where $\{q_j\}_{j=1,2}$ is a probability distribution and $\{\sigma_j\}_{j=1,2}$ are two PDP states. The generalization to a larger number of components is straightforward.

Closely following von Neumann, we apply a set of $\rho$-separating SPMs, corresponding to the repeatable MPP instrument $\{e_j,\sigma_j\}_{j\in\{1,2\}}$ with effects such that $e_j(\sigma_{j'})=\delta_{jj'}$. This is physically modeled by two SPMs with opposite mechanical behaviors: if one SPM is transparent for, say, $\sigma_1$, the other is transparent for $\sigma_2$. After the separation, which is done isothermally, the two species ($\sigma_1$ and $\sigma_2$) are contained in two separate chambers, both of the same size as the initial chamber.
The number of particles in a state $\sigma_j$ is $e_j(\rho)N=q_jN$.
In agreement with our preceding discussion and previous analyses~\cite{von1955mathematical,hanggi2013violation,krumm2017thermodynamics}, we can assume that this first step of the separation is basically a solid translation of coordinates, so that the work worth of this step is zero.

The two SPMs are then replaced by impermeable walls and we isothermally compress the chamber with $\sigma_1$ (resp., $\sigma_2$) until its volume becomes $q_1V$ (resp., $q_2V$), so that after the compression both chambers have the same (initial) pressure.
From the ideal gas law, the amount of work needed for the isothermal compression c$\to$d is
\begin{align}
&-\int_{V}^{q_1V}\frac{q_1NkT}{V'}dV'-\int_{V}^{q_2V}\frac{q_2NkT}{V'}dV'\nonumber\\
	&=H(\{q_j\})NkT\ln 2\;,\label{eq:work-in}
\end{align}
where $T$ is the temperature of the environment, $k$ is the Boltzmann constant, and $H(\{q_j\}):=-\sum_j q_j\ln q_j$. This ends the separation protocol. If the decomposition of $\rho$ has more than two PDP states, we repeat this protocol for each perfectly distinguishable state and obtain the same result.

Next, we consider the mixing process. Here we deviate from von Neumann's process, in that we want to go back to the initial mixed state, whereas von Neumann's final state is pure. After the separation process has been completed, we are in the situation in which the two (pure, distinguishable) species $\sigma_1$ and $\sigma_2$ are in two separate chambers with equal pressure (i.e., the initial pressure). The total volume of the two chambers together equals the initial volume, but now we know which species is present in each chamber.
If we were to replace the wall between the two chambers again with the same SPMs used during the separation step, by letting the two species slowly and independently expand, we would be able to gain back the work invested in the separation stage with the gas restored to its initial state. We would then have achieved a trivial, reversible cycle.

\begin{figure}[htbp]
	\centering
	\includegraphics[width=8cm]{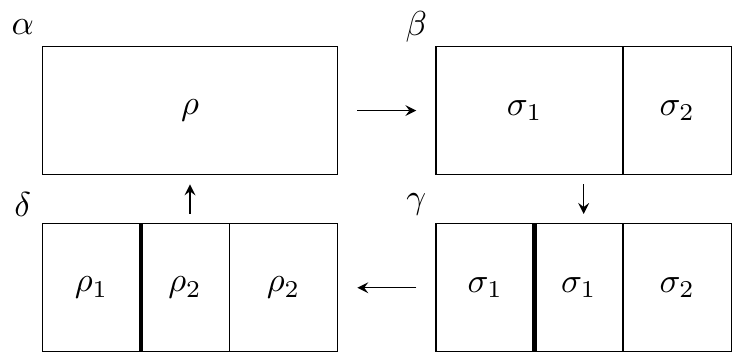}
	\caption{The cycle used in the proof of Theorem \ref{theorem:main}. \textbf{[$\alpha\to\beta$]}: work is invested to achieve the separation as in Fig.~\ref{fig:von_neumann}. \textbf{[$\beta\to\gamma$]}: additional partition plates are inserted at proper positions (zero work). \textbf{[$\gamma\to\delta$]}: to each chamber, a suitable reversible transformation, with zero work-cost, is applied and the states therein contained are transformed into $\rho_i$. \textbf{[$\delta\to\alpha$]}: work is extracted, by running the separation process for the decomposition $\rho=\sum_ip_i\rho_i$ backwards.}
	\label{figure:main}
\end{figure}

Instead, we consider a \textit{different} decomposition of $\rho$ in pure distinguishable states, as assumption A1 does not exclude such a possibility. Let us denote the alternative decomposition of $\rho$ by $\sum_ip_i\rho_i$. As depicted in Fig.~\ref{figure:main}, we can freely insert additional partitions in a suitable way, and transform the pure states (which are known) in each partition into $\rho_1$ or $\rho_2$, as needed. In this way (by removing additional walls as necessary), we have transformed the arrangement corresponding to the decomposition $\sum_jq_j\sigma_j$ into the arrangement corresponding to $\sum_ip_i\rho_i$ without the need to account for any new term in the thermodynamic balance, thanks to assumption A2.

Then, by using two new SPMs tuned to match the new PDP decomposition $\sum_ip_i\rho_i$, and by following the separation steps backward, we can bring the system back to its original state having gained in the process (isothermal expansion) an amount of work equal to:
\begin{equation}
    H(\{p_i\})NkT\ln2 .\label{eq:W_mix}
\end{equation}
As a whole, therefore, the amount of work we can extract from this cycle is calculated from \eqref{eq:work-in} and \eqref{eq:W_mix} as follows:
\begin{align}
    \Delta W=[H(\{p_i\})-H(\{q_j\})]\;NkT\ln 2\;.\label{eq:delta_W}
\end{align}

We now invoke, following von Neumann, the second law of thermodynamics, which implies that $\Delta W$ cannot be strictly positive; otherwise, we would have constructed a \textit{perpetuum mobile} of the second kind. Moreover, since the same must hold also if we exchange the role of the two decompositions, we conclude that $\Delta W$ must be exactly zero. But this must hold for \textit{any} PDP decomposition of $\rho$. Hence we obtain the following:
\begin{theorem}\label{theorem:main}
    Under assumptions A1 and A2, the second law of thermodynamics implies that for all weakly spectral $\rho\in\Omega$ we have
    \begin{equation}
        H(\{p_i\})=H(\{q_i\})\;,\label{eq:ent}
    \end{equation}
	for all $\{p_i\},\{q_i\}\in\mathcal{D}(\rho)$. In other words, a necessary condition for the validity of the second law is that any weakly spectral state has a unique spectral entropy.
\end{theorem}
Indeed, we provide an explicit example of a GPT where the spectral entropy is not uniquely defined in Appendix \ref	{sec:app-different-spectral-entropies}.

Here we derive concavity of the spectral entropy under the assumption that \textit{all} the states of the theory are weakly spectral and that the second law is valid.


\begin{theorem}[Concavity]\label{theorem:concavity}
Let $\rho_1,\rho_2\in\Omega$ be states. Then we have the following inequality:
\begin{equation}
    H\left(p\rho_1+(1-p)\rho_2\right)\ge pH(\rho_1)+(1-p)H(\rho_2)\;,
\end{equation}
for any value $0\le p\le 1$.
\end{theorem}
\begin{proof}
Consider a state $\rho$ whose decomposition into PDP states is $\rho=\sum_i q_i \sigma_i$. First, as shown in Fig.~\ref{fig:concavity}, we separate it into its pure components using a suitable set of $\rho$-separating SPMs. The resulting arrangement is shown as (b) in Fig.~\ref{fig:concavity} and costs an amount of work proportional to $H(\rho)$. Then, by adding partitions, transforming pure states, and mixing, we can arrive at the arrangement shown in (c) of Fig.~\ref{fig:concavity}. In this step we can gain an amount of work proportional to $\sum_j p_j H(\rho_j)$, where $p_1\rho_1+p_2\rho_2=\rho$.
Finally, since $\rho$ is the convex combination of $\rho_1$ and $\rho_2$, we can accomplish the process (c)$\to$(a) by only removing the partion, without requiring any work. Therefore, the work we can extract from the isothermal cycle (a)$\to$(b)$\to$(c)$\to$(a) is $\Delta W=\left[-H(\rho)+\sum_j p_jH(\rho_j)\right]NkT\ln 2$. The second law implies that $\Delta W\le 0$, which completes the proof.

\begin{figure}[htbp]
	\centering
	\includegraphics[width=8.5cm]{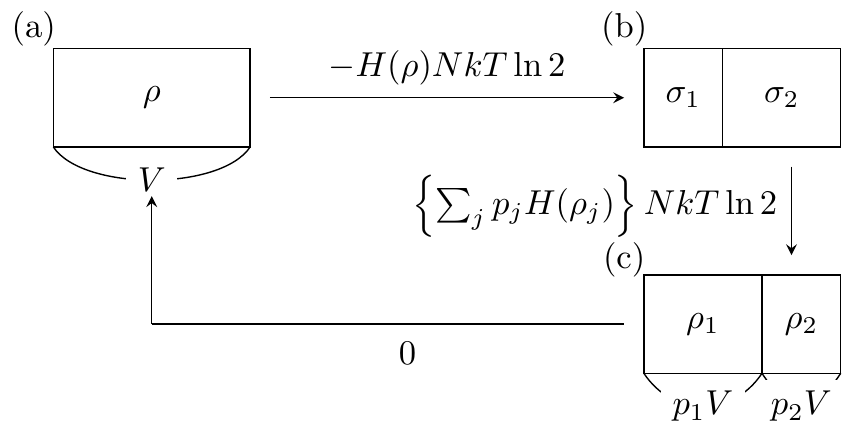}
	\caption{The proof of the concavity of spectral entropy. \textbf{[(a)$\to$(b)]}: separate $\rho$ into PDP states; this process needs work $H(\rho)NkT\ln 2$. \textbf{[(b)$\to$(c)]}: create $\rho_1$ and $\rho_2$ from pure states; the amount of work $\sum_i p_iH(\rho_i)NkT\ln 2$ can be extracted from this process. \textbf{[(c)$\to$(a)]}: since $\rho$ is the convex combination of $\rho_1$ and $\rho_2$, we can accomplish this process by just removing the wall with no work.}
	\label{fig:concavity}
\end{figure}
\end{proof}

\section{Groenewold--Ozawa information gain}
The uniqueness of the spectral entropy functional can be used to extend the definition of the Groenewold--Ozawa information gain~\cite{groenewold1971problem,ozawa1986information} to GPTs that satisfy weak spectrality as follow:

\begin{definition}
For any state $\rho\in\Omega$ and instrument $s$,
we define the Groenewold--Ozawa information gain as follows:
\begin{equation}
    I_{\mathrm{G}}(\rho,s):=H(\rho)-\sum_je_j(\rho)H\left(\frac{s_j(\rho)}{e_j(\rho)}\right)\;,
\end{equation}
where $e_j(\rho):=(u\circ s_j)(\rho)$ is the probability of the $j$-th outcome.
\end{definition}

The Groenewold--Ozawa information gain earns an operational interpretation in terms of the preorder introduced in Definition~\ref{def:rho-ref} as a consequence of the following fact, proved in Appendix \ref{sec:app-goig}.

\begin{theorem}\label{theorem:monotone}
   For any state $\rho\in\Omega$ and instruments $t$ and $s$,
   $t\succ_\rho s$ implies the following inequality:
    \begin{equation}
        I_{\mathrm{G}}(\rho,t)\ge I_{\mathrm{G}}(\rho,s)\;.\label{eq:Groenewold-monotone}
    \end{equation}
\end{theorem}

Recently, Ref.~\cite{danageozian2021thermodynamic} proved a relation between the Groenewold--Ozawa information gain and the heat absorbed by the system during a measurement process. This result together with our inequality \eqref{eq:Groenewold-monotone} suggests a link between the thermodynamic and the resource-theoretic characterization of quantum measurements. We leave this point open for future studies.

\section{Conclusions}
In this work, we have shown how von Neumann's thought experiment can be formulated in a purely operational language, without resorting to the structure of Hilbert spaces or the spectral theorem. An advantage of our reformulation is that we can now appreciate how important it is, in von Neumann's argument, to assume the validity of the second law of thermodynamics from the beginning. It is the second law, and not the spectral theorem, to force the entropy to be unique, lest we build a \textit{perpetuum mobile} of the second kind. In this sense, the second law is used by von Neumann as a consistency check, a first principle of logic rather than a law of physics~\cite{watanabe55,watanabe65,buscemi2021fluctuation,aw-buscemi-scarani}.

A problem left open is that of finding relations between the spectral entropy considered in this work and other entropies that can be considered in GPTs~\cite{short2010entropy,barnum2010entropy,kimura2010distinguishability,perinotti2021shannon}. Moreover, since we can regard von Neumann's device as a process for extracting work from measurements, there may be a close relationship between the present discussion and previous research on work extraction from quantum measurement processes~\cite{hayashi2017measurement}. We leave these questions for future works.

\section{Acknowledgments}
The authors are very grateful to Michele Dall'Arno, Masanao Ozawa, and Mark M. Wilde for their useful comments.
Support from MEXT Quantum Leap Flagship Program (MEXT QLEAP) Grant No. JPMXS0120319794 is acknowledged.
S.M. would like to take this opportunity to thank the “Nagoya University Interdisciplinary Frontier Fellowship” supported by JST and Nagoya University.
H.A. is supported by a JSPS Grant-in-Aids for JSPS Fellows No. JP22J14947, a JSPS Grant-in-Aids for Scientific Research (B) Grant No. JP20H04139, and a Grant-in-Aid for JST SPRING No. JPMJSP2125.
F.B. acknowledges support from MEXT Quantum Leap Flagship Program (MEXT QLEAP) Grant No. JPMXS0120319794; from MEXT-JSPS Grant-in-Aid for Transformative Research Areas (A) “Extreme Universe”, No. 21H05183; from JSPS KAKENHI Grants No. 19H04066 and 20K03746.

\medskip\textit{Author contributions.}---S.~M. and H.~A. contributed to this work equally.

\appendix
\section{Proof of Lemma \ref{lemma:f-g-instrument}}\label{app:lamma:f-g}
Firstly we prove that the output states of fine-grained instruments are pure. Consider a state space $\Omega$ and a state $\rho\in\Omega$. Let $s=\{s_j\}$ be a fine-grained instrument and $e_j(\rho)$ denotes $(u\circ s_j)(\rho)$ for simplicity.
Without loss of generality, suppose that a post-measurement state corresponding to the outcome $j=|J|$,  
\begin{equation}
\sigma_{|J|}(\rho):=\frac{1}{e_{|J|}(\rho)}s_{|J|}(\rho)    
\end{equation}
is not pure, that is, there is a convex decomposition like 
\begin{equation}
\sigma_{|J|}(\rho)=\sum_{l\in L} q_l \sigma_l,
\end{equation}
where $\{q_l\}_{l\in L}$ is a probability distribution and $|L|>1$. Note that both $\{q_l\}_{l\in L}$ and $\sigma_l$ depend on $\rho$. 

If we multiply both side of this equation by $e_{|j|}(\rho)$, we obtain 
\begin{equation}
    s_{|J|}(\rho)=\sum_{l\in L}e_{|J|}(\rho)q_l\sigma_l.
\end{equation}
Let us now introduce events $\tilde{s}_{l,|J|}$ such that
\begin{equation}
    \tilde{s}_{l,|J|}(\rho)=e_{|J|}(\rho)q_l\sigma_l.
\end{equation}

Let $K^J$ and $K^L$ be sets such that
\begin{align}
    K_1&:=\{1,\dots,j,\dots,|J|-1\}\\
    K_2&:=\{(1,|J|),\dots,(l,|J|),\dots,(|L|,|J|)\},
\end{align}
where $j\in J\setminus\{|J|\}$ and $l\in L$.

Now we define an index $k\in K$ as a direct sum of $K^J$ and $K^L$:
\begin{equation}
    K:=K_1\cup K_2
\end{equation}
There is a family of pure states $\{t_k(\rho)\}_{k\in K}$ such that
\begin{equation}
    t_k(\rho):=
    \begin{cases}
        &s_k(\rho)\quad \mathrm{if}\:k\in K_1\\
        &\tilde{s}_{k}(\rho)\quad \mathrm{if}\:k\in K_2
    \end{cases}
\end{equation}
Since we have $s_{|J|}(\rho)=\sum_{k\in K_2}\tilde{s}_k(\rho)$, we have $t\succ_\rho s$. However, since the $|j|$-th output state of $s$ is mixed, and it is not possible to make a mixed state pure by further convex mixtures, it is clear that $s\nsucc_\rho t$. This contradicts the assumption that $s$ is fine-grained, thus proving the first part, that is, the states $\frac{1}{e_j(\rho)}s_j(\rho)$ must be pure for all outcomes $j$ and all initial states $\rho$.

Now we show that fine-grained instruments are MPP instruments. Suppose that a state $\rho\in\Omega$ has a convex decomposition $\rho=p\rho_1+(1-p)\rho_2$, where $\rho_1,\rho_2$ are two different state on $\Omega$. From the affinity of $s_j$, $s_j(\rho)=ps_j(\rho_1)+(1-p)s_j(\rho_2)$ holds. For what we said before, $s_j(\rho)$ is proportional to a pure state. Therefore, it must be that both $s_j(\rho_1)$ and $s_j(\rho_2)$ are proportional to $s_j(\rho)$. This means that the post-measurement state does not depend on the initial state, if not from the outcome $j$, which implies that $s$ is MPP.

\section{A theory satisfying A2 but not strong symmetry}\label{sec:app2}
Here we give an explicit example of a theory that satisfies property A2 above, that is, the existence of free pure-state transformations, but does not satisfy strong symmetry.

Let $\cH$ be a finite-dimensional Hilbert space and let $\cS$ be a state space on $\cH$, that is, a convex set of positive semidefinite matrices on $\cH$ with unit trace.
Also, let $\cP\subseteq\cS$ be a set of rank one matrices, which corresponds to the set of pure states. 
Consider two systems $\cH_A$ and $\cH_B$ and let $\cP_{\rA}\otimes\cP_{\rB}$ be the set of product pure states. The state space $\Omega$ of the theory is the convex hull $\mathrm{SEP(A;B)}$ of $\cP_{\rA}\otimes\cP_{\rB}$.

For a state space $\Omega$, define a class of transformation $\cF(\Omega)$ as the set of all linear maps $F$ satisfying $F(\Omega)=\Omega$. Any reversible transformation clearly belongs to $\cF(\Omega)$.

\begin{definition}[$k$-symmetry]
	We say that a state space $\Omega$ is $k$-symmetric if there exists a map
	$F\in\cF(\Omega)$ such that $\rho_i = F(\sigma_i)$, for $i = 1, \cdots, k$, for any pair of $k$-tuples of perfectly distinguishable pure states $\{\rho_i\}_{i=1}^k$ and $\{\sigma_i\}_{i=1}^k$.
\end{definition}

Now we show the difference between strong symmetry and weak symmetry by giving the following counterexample.
\begin{theorem}\label{theorem:sep}
	$\mathrm{SEP(A;B)}$ is $1$-symmetric but not $2$-symmetric.
\end{theorem}

We invoke the following two results. The first result is about the necessary and sufficient condition for two states in SEP to be perfectly distinguished \cite[Theorem 2.4]{arai2019perfect}.
Notice that non-orthogonal states can be perfectly distinguished in SEP because the set of all measurements in SEP is larger than the set of bipartite POVMs.
\begin{lemma}\label{lemma:sep_perf}
	In $\mathrm{SEP(A; B)}$, two pure states $\rho_1 = \rho^{\rA}_1\otimes\rho^{\rB}_1$ and $\rho_2=\rho^{\rA}_2\otimes\rho^{\rB}_2$ are perfectly distinguishable if and only if they satisfy
	\begin{equation}\label{eq:sep-dist}
		\Tr\rho^{\rA}_1\rho^{\rA}_2+\Tr\rho^{\rB}_1\rho^{\rB}_2\le 1.
	\end{equation}
\end{lemma}
The second result gives the form of the transformation maps in $\cF(\mathrm{SEP(A;B)})$ concretely \cite[Theorem 3]{friedland2011automorphism}.
\begin{lemma}\label{lemma:sep_transf}
	For the linear map $F$ from $\cT(\cH_{\rA}\otimes\cH_{\rB})\to\cT(\cH_{\rA}\otimes\cH_{\rB})$, the following are equivalent:
	\begin{enumerate}[(i)]
		\item $F\in\cF(\mathrm{SEP(A;B)})$.
		\item $F(\cP_{\rA}\otimes\cP_{\rB})=\cP_{\rA}\otimes\cP_{\rB}$.
		\item $F(A\otimes B)=F_{\rA}(A)\otimes F_{\rB}(B)$, or $\dim\cH_{\rA}=\dim\cH_{\rB}$ and $F(A\otimes B)=F_{\rB}(B)\otimes F_{\rA}(A)$, where $F_{\rA}(A)=U_{\rA} A U_{\rA}^\dagger$ or $U_{\rA} A^{\mathsf{T}}U_{\rA}^\dagger$ and $F_{\rB}(B)=V_{\rB}BV_{\rB}^\dagger$ or $V_{\rB}B^{\mathsf{T}}V_{\rB}^\dagger$.
	\end{enumerate}
\end{lemma}
\begin{proof}[Proof of Theorem \ref{theorem:sep}]
	Since $\cF(\mathrm{SEP(A;B)})$ contains all local unitary maps, $\mathrm{SEP(A;B)}$ clearly satisfies $1$-symmetry.
	
	Now, we show that $\mathrm{SEP(A;B)}$ is not $2$-symmetric by giving a counterexample.
	Take the following four separable pure states:
	\begin{align}
		\rho_1&=\rho_1^{\mathrm{A}}\otimes\rho_1^{\mathrm{B}}=
		\begin{bmatrix}
			1 & 0\\
			0 & 0
		\end{bmatrix}\otimes
		\begin{bmatrix}
			1 & 0\\
			0 & 0
		\end{bmatrix},\label{eq:ex_sep1}\\
		\rho_2&=\frac{1}{2}
		\begin{bmatrix}
			1 & 1\\
			1 & 1
		\end{bmatrix}\otimes\frac{1}{2}
		\begin{bmatrix}
			1 & 1\\
			1 & 1
		\end{bmatrix},\label{eq:ex_sep2}\\
		\sigma_1&=
		\begin{bmatrix}
			1 & 0\\
			0 & 0
		\end{bmatrix}\otimes
		\begin{bmatrix}
			1 & 0\\
			0 & 0
		\end{bmatrix},\label{eq:ex_sep3}\\
		\sigma_2&={}
		\begin{bmatrix}
			0 & 0\\
			0 & 1
		\end{bmatrix}\otimes
		\begin{bmatrix}
			0 & 0\\
			0 & 1
		\end{bmatrix}\label{eq:ex_sep4}.
	\end{align}
	By direct inspection, we can verify that the two dichotomies $\{\rho_1,\rho_2\}$ and $\{\sigma_1,\sigma_2\}$ both satisfy condition~\eqref{eq:sep-dist} and thus, by Lemma \ref{lemma:sep_perf}, both contain perfectly distinguishable pure states in $\mathrm{SEP(A;B)}$.
	
	Assume that there is a map $F\in\cF(\mathrm{SEP}(A;B))$ where $\sigma_1=F(\rho_1)$ and $\sigma_2=F(\rho_2)$.
	From Lemma \ref{lemma:sep_transf}, the following equality should holds:
	\begin{equation}
		\begin{split}
			\Tr\{\sigma_1\sigma_2\}&=\Tr\{F_{\rA}(\rho_1^{\rA})F_{\rA}(\rho_2^{\rA})\otimes F_{\rB}(\rho_1^{\rB})F_{\rB}(\rho_2^{\rB})\}\\
			&=\Tr\{\rho_1^{\rA}\rho_2^{\rA}\otimes\rho_1^{\rB}\rho_2^{\rB}\}\\
			&=\Tr\{\rho_1\rho_2\}.\label{eq:sep_tr}
		\end{split}
	\end{equation}
	However, now we have
	\begin{equation}
		\Tr\{\rho_1\rho_2\}=\frac{1}{4},\quad \Tr\{\sigma_1\sigma_2\}=0.
	\end{equation}
	This contradicts \eqref{eq:sep_tr}.
	Thus $\mathrm{SEP}(A;B)$ is $1$-symmetric but not $2$-symmetric.
\end{proof}

\section{Non-uniqueness of spectral entropy}\label{sec:app-different-spectral-entropies}

Theorem 1 says that the entropies corresponding to different PDP decompositions must all be equal if the second law is valid. However, in general, there exist theories where the same state can have decompositions with different entropies. As a consequence of Theorem 1, all such theories are intrinsically incompatible with the second law.

We consider again the state space SEP.
Lemma~\ref{lemma:sep_perf} implies that the following two separable states are perfectly distinguishable:
\begin{align}
    \rho_1
    &=
    \begin{bmatrix}
        1 & 0\\
        0 & 0
    \end{bmatrix}
    \otimes
    \begin{bmatrix}
        1 & 0\\
        0 & 0
    \end{bmatrix},\\
    \rho_2
    &=
    \cfrac{1}{2}\begin{bmatrix}
        1 & 1\\
        1 & 1
    \end{bmatrix}
    \otimes
    \cfrac{1}{2}\begin{bmatrix}
        1 & 1\\
        1 & 1
    \end{bmatrix}.
\end{align}
Besides, the reference \cite{arai2019perfect} gives the following measurement $\{e_1,e_2\}$ that discriminate $\{\rho_1,\rho_2\}$ perfectly:
\begin{align}\label{eq:meas-1}
    e_1(\rho)
    &=\Tr \left\{
    \cfrac{1}{2}
    \begin{bmatrix}
    	2 & 0 & 0 & -1\\
    	0 & 0 & -1 & 0\\
    	0 & -1 & 0 & 0\\
    	-1 & 0 & 0 & 2
    \end{bmatrix}
    \rho
    \right\},\\
    e_2(\rho)\label{eq:meas-2}
    &=\Tr \left\{
    \cfrac{1}{2}
    \begin{bmatrix}
    	0 & 0 & 0 & 1\\
    	0 & 2 & 1 & 0\\
    	0 & 1 & 2 & 0\\
    	1 & 0 & 0 & 0
    \end{bmatrix}
    \rho
    \right\}.
\end{align}
The two matrices appearing above have negative eigenvalues: this is because SEP does not contain any entangled state.

Next,
we extend the state space of SEP slightly.
Consider the following density matrices with unit rank:
\begin{align}
	\sigma_1
	&=\cfrac{1}{6}
	\begin{bmatrix}
    	3 & \sqrt{3} & \sqrt{3} & \sqrt{3}\\
    	\sqrt{3} & 1 & 1 & 1\\
    	\sqrt{3} & 1 & 1 & 1\\
    	\sqrt{3} & 1 & 1 & 1
    \end{bmatrix},\\
    \sigma_2
    &=\cfrac{1}{6}
    \begin{bmatrix}
    	3 & -\sqrt{3} & -\sqrt{3} & -\sqrt{3}\\
    	-\sqrt{3} & 1 & 1 & 1\\
    	-\sqrt{3} & 1 & 1 & 1\\
    	-\sqrt{3} & 1 & 1 & 1
    \end{bmatrix}.
\end{align}
Because $\sigma_1$ and $\sigma_2$ are not separable,
$\sigma_1,\sigma_2\not\in \mathrm{SEP}(A;B)$.
Then consider the following state space $\overline{\Omega}$:
\begin{align}
	\overline{\Omega}:=
	\mathrm{Hul}\left(\mathrm{SEP}(A;B)\cup\{\sigma_1,\sigma_2\}\right),
\end{align}
where $\mathrm{Hul}(X)$ denotes the convex hull of $X$.
We remark that $\rho_i$ and $\sigma_j$ are pure because they are rank 1 matrices.

In the model corresponding to this state space,
the set of all measurements is given as the set of all affine functions $\{e_j\}_{j\in J}$ satisfying $e_j(\rho)\ge0$ and $\sum_je_j(\rho)=1$, for any $\rho\in\overline{\Omega}$ and $j\in J$.
Because the state space $\overline{\Omega}$ is smaller than the set of all density matrices, but larger than SEP(A;B),
the set of all measurements is larger than the set of POVMs and smaller than the set of measurements in SEP. In particular,
because $e_j(\sigma_i)\ge0$ for all $i,j$,
the measurement $\{e_1,e_2\}$ appearing in Eqs.~\eqref{eq:meas-1} and~\eqref{eq:meas-2} is allowed in the model.
\color{black}
Because the two states $\sigma_1,\sigma_2$ are orthogonal quantum states, they can be perfectly distinguished in conventional quantum theory and, therefore, are perfectly distinguishable also in this extended model.

This implies that the state $\rho:=\cfrac{1}{3}\rho_1+\cfrac{2}{3}\rho_2$ can be decomposed into PDP states in two different ways, as follows:
\begin{align}
	\rho&=\frac{1}{3}\rho_1+\frac{2}{3}\rho_2,\\
	&=\frac{3+\sqrt{3}}{6}\sigma_1+\frac{3-\sqrt{3}}{6}\sigma_2\;,
\end{align}
which clearly possess two different values for the spectral entropy.

\section{The proof of Theorem \ref{theorem:monotone}}\label{sec:app-goig}
Consider two instruments $s=\{s_j\}_{j\in J}$ and $t=\{t_k\}_{k\in K}$, where $e_j(\rho)=(u\circ s_j)(\rho)$ and $f_k(\rho)=(u\circ t_k)(\rho)$.
Suppose that $t\succ_\rho s$ holds.
Then we have
\begin{equation}
	\forall j\in J,\quad s_j(\rho)=\sum_{k\in K}p(j|k)t_k(\rho)
\end{equation}
where $\sum_{j\in J}p(j|k)=1$ holds for all $k\in K$. 

Firstly, we have
\begin{equation}
	\begin{split}
		e_j(\rho)&=(u\circ s_j)(\rho)=\left\{u\circ \left(\sum_{k\in K}p(j|k)t_k\right)\right\}(\rho)\\
		&=\sum_{k\in K}p(j|k)(u\circ t_k)(\rho)=\sum_{k\in K}p(j|k)f_k(\rho).\label{eq:e-f}
	\end{split}
\end{equation}
The third equality is because of the affinity of $u$.
Therefore, we have
\begin{equation}
	\begin{split}
		\frac{s_j(\rho)}{e_j(\rho)}&=\frac{\sum_{k\in K} p(j|k)t_k(\rho)}{\sum_{k\in K} p(j|k)f_k(\rho)}\\
		&=\sum_{k\in K}\frac{p(j|k)f_k(\rho)}{\sum_{k\in K} p(j|k)f_k(\rho)}\frac{t_k(\rho)}{f_k(\rho)}.
	\end{split}
\end{equation}
This implies that the state $s_j(\rho)/e_j(\rho)$ is the convex combination of states $t_k(\rho)/f_k(\rho)$ with the probability $p(j|k)f_k(\rho)/\sum_{k\in K}p(j|k)f_k(\rho)$.

If we use the result of Theorem \ref{theorem:monotone} and \eqref{eq:e-f}, we have
\begin{equation}
	\begin{split}
		&\sum_{j\in J} e_j(\rho)H\left(\frac{s_j(\rho)}{e_j(\rho)}\right)\\
		&\ge \sum_{j\in J} e_j(\rho)\sum_{k\in K} \frac{p(j|k)f_k(\rho)}{\sum_{k\in K}p(j|k)f_k(\rho)}H\left(\frac{t_k(\rho)}{f_k(\rho)}\right)\\
		&=\sum_{j\in J}e_j(\rho)\sum_{k\in K}\frac{p(j|k)f_k(\rho)}{e_j(\rho)}H\left(\frac{t_k(\rho)}{f_k(\rho)}\right)\\
		&=\sum_{j\in J}\sum_{k\in K}p(j|k)f_k(\rho)H\left(\frac{t_k(\rho)}{f_k(\rho)}\right)=\sum_{k\in K} f_k(\rho)H\left(\frac{t_k(\rho)}{f_k(\rho)}\right).
	\end{split}
\end{equation}

Therefore, we obtain
\begin{equation}
\begin{split}
    &H(\rho)-\sum_{j\in J} e_j(\rho)H\left(\frac{s_j(\rho)}{e_j(\rho)}\right)\\
    &\le H(\rho)-\sum_{k\in K}f_k(\rho)H\left(\frac{t_k(\rho)}{f_k(\rho)}\right),
\end{split}
\end{equation}
which is the desired inequality.

\bibliography{myref}

\begin{thebibliography}{47}%
\makeatletter
\providecommand \@ifxundefined [1]{%
 \@ifx{#1\undefined}
}%
\providecommand \@ifnum [1]{%
 \ifnum #1\expandafter \@firstoftwo
 \else \expandafter \@secondoftwo
 \fi
}%
\providecommand \@ifx [1]{%
 \ifx #1\expandafter \@firstoftwo
 \else \expandafter \@secondoftwo
 \fi
}%
\providecommand \natexlab [1]{#1}%
\providecommand \enquote  [1]{``#1''}%
\providecommand \bibnamefont  [1]{#1}%
\providecommand \bibfnamefont [1]{#1}%
\providecommand \citenamefont [1]{#1}%
\providecommand \href@noop [0]{\@secondoftwo}%
\providecommand \href [0]{\begingroup \@sanitize@url \@href}%
\providecommand \@href[1]{\@@startlink{#1}\@@href}%
\providecommand \@@href[1]{\endgroup#1\@@endlink}%
\providecommand \@sanitize@url [0]{\catcode `\\12\catcode `\$12\catcode
  `\&12\catcode `\#12\catcode `\^12\catcode `\_12\catcode `\%12\relax}%
\providecommand \@@startlink[1]{}%
\providecommand \@@endlink[0]{}%
\providecommand \url  [0]{\begingroup\@sanitize@url \@url }%
\providecommand \@url [1]{\endgroup\@href {#1}{\urlprefix }}%
\providecommand \urlprefix  [0]{URL }%
\providecommand \Eprint [0]{\href }%
\providecommand \doibase [0]{https://doi.org/}%
\providecommand \selectlanguage [0]{\@gobble}%
\providecommand \bibinfo  [0]{\@secondoftwo}%
\providecommand \bibfield  [0]{\@secondoftwo}%
\providecommand \translation [1]{[#1]}%
\providecommand \BibitemOpen [0]{}%
\providecommand \bibitemStop [0]{}%
\providecommand \bibitemNoStop [0]{.\EOS\space}%
\providecommand \EOS [0]{\spacefactor3000\relax}%
\providecommand \BibitemShut  [1]{\csname bibitem#1\endcsname}%
\let\auto@bib@innerbib\@empty
\bibitem [{\citenamefont {Maxwell}(1871)}]{maxwell1871theory}%
  \BibitemOpen
  \bibfield  {author} {\bibinfo {author} {\bibfnamefont {J.~C.}\ \bibnamefont
  {Maxwell}},\ }\href@noop {} {\emph {\bibinfo {title} {Theory of heat}}}\
  (\bibinfo  {publisher} {Appleton, London},\ \bibinfo {year}
  {1871})\BibitemShut {NoStop}%
\bibitem [{\citenamefont {von Neumann}(1955)}]{von1955mathematical}%
  \BibitemOpen
  \bibfield  {author} {\bibinfo {author} {\bibfnamefont {J.}~\bibnamefont {von
  Neumann}},\ }\href@noop {} {\emph {\bibinfo {title} {Mathematical foundations
  of quantum mechanics}}}\ (\bibinfo  {publisher} {Princeton University Press,
  Princeton, NJ},\ \bibinfo {year} {1955})\BibitemShut {NoStop}%
\bibitem [{\citenamefont {Einstein}(1996)}]{einstein1914beitrage}%
  \BibitemOpen
  \bibfield  {author} {\bibinfo {author} {\bibfnamefont {A.}~\bibnamefont
  {Einstein}},\ }\bibfield  {title} {\bibinfo {title} {Contributions to
  {Q}uantum {T}heory},\ }in\ \href@noop {} {\emph {\bibinfo {booktitle} {The
  {C}ollected {P}apers of {A}lbert {E}instein; {V}olume 6 the {B}erlin years:
  {W}ritings 1914-1917}}},\ \bibinfo {editor} {edited by\ \bibinfo {editor}
  {\bibfnamefont {A.~J.}\ \bibnamefont {Kox}}, \bibinfo {editor} {\bibfnamefont
  {M.~J.}\ \bibnamefont {Klein}},\ and\ \bibinfo {editor} {\bibfnamefont
  {R.}~\bibnamefont {Schulmann}}}\ (\bibinfo  {publisher} {Princeton University
  Press, Princeton, NJ},\ \bibinfo {year} {1996})\ pp.\ \bibinfo {pages}
  {29--40}\BibitemShut {NoStop}%
\bibitem [{\citenamefont {Klein}(1967)}]{klein1967thermodynamics}%
  \BibitemOpen
  \bibfield  {author} {\bibinfo {author} {\bibfnamefont {M.~J.}\ \bibnamefont
  {Klein}},\ }\bibfield  {title} {\bibinfo {title} {Thermodynamics in
  einstein's thought},\ }\href {https://doi.org/10.1126/science.157.3788.509}
  {\bibfield  {journal} {\bibinfo  {journal} {Science}\ }\textbf {\bibinfo
  {volume} {157}},\ \bibinfo {pages} {509} (\bibinfo {year}
  {1967})}\BibitemShut {NoStop}%
\bibitem [{\citenamefont {Szilard}(1929)}]{szilard1929uber}%
  \BibitemOpen
  \bibfield  {author} {\bibinfo {author} {\bibfnamefont {L.}~\bibnamefont
  {Szilard}},\ }\bibfield  {title} {\bibinfo {title} {{\"u}ber die
  {E}ntropieverminderung in einem thermodynamischen {S}ystem bei {E}ingriffen
  intelligenter {W}esen},\ }\href
  {https://doi.org/https://doi.org/10.1007/BF01341281} {\bibfield  {journal}
  {\bibinfo  {journal} {Zeitschrift f{\"u}r Physik}\ }\textbf {\bibinfo
  {volume} {53}},\ \bibinfo {pages} {840} (\bibinfo {year} {1929})}\BibitemShut
  {NoStop}%
\bibitem [{\citenamefont {Szilard}(1964)}]{szilard1964decrease}%
  \BibitemOpen
  \bibfield  {author} {\bibinfo {author} {\bibfnamefont {L.}~\bibnamefont
  {Szilard}},\ }\bibfield  {title} {\bibinfo {title} {On the decrease of
  entropy in a thermodynamic system by the intervention of intelligent
  beings},\ }\href {https://doi.org/https://doi.org/10.1002/bs.3830090402}
  {\bibfield  {journal} {\bibinfo  {journal} {Behavioral Science}\ }\textbf
  {\bibinfo {volume} {9}},\ \bibinfo {pages} {301} (\bibinfo {year}
  {1964})}\BibitemShut {NoStop}%
\bibitem [{\citenamefont {{D'Ariano}}\ \emph {et~al.}(2017)\citenamefont
  {{D'Ariano}}, \citenamefont {Chiribella},\ and\ \citenamefont
  {Perinotti}}]{dariano-OPT}%
  \BibitemOpen
  \bibfield  {author} {\bibinfo {author} {\bibfnamefont {G.~M.}\ \bibnamefont
  {{D'Ariano}}}, \bibinfo {author} {\bibfnamefont {G.}~\bibnamefont
  {Chiribella}},\ and\ \bibinfo {author} {\bibfnamefont {P.}~\bibnamefont
  {Perinotti}},\ }\href@noop {} {\emph {\bibinfo {title} {Quantum Theory from
  First Principles: An Informational Approach}}}\ (\bibinfo  {publisher}
  {Cambridge University Press, Cambridge},\ \bibinfo {year} {2017})\BibitemShut
  {NoStop}%
\bibitem [{\citenamefont {Janotta}\ and\ \citenamefont
  {Hinrichsen}(2014)}]{janotta2014generalized}%
  \BibitemOpen
  \bibfield  {author} {\bibinfo {author} {\bibfnamefont {P.}~\bibnamefont
  {Janotta}}\ and\ \bibinfo {author} {\bibfnamefont {H.}~\bibnamefont
  {Hinrichsen}},\ }\bibfield  {title} {\bibinfo {title} {Generalized
  probability theories: what determines the structure of quantum theory?},\
  }\href {https://doi.org/https://doi.org/10.1088/1751-8113/47/32/323001}
  {\bibfield  {journal} {\bibinfo  {journal} {Journal of Physics A:
  Mathematical and Theoretical}\ }\textbf {\bibinfo {volume} {47}},\ \bibinfo
  {pages} {323001} (\bibinfo {year} {2014})}\BibitemShut {NoStop}%
\bibitem [{\citenamefont {Pl{\'a}vala}(2021)}]{plavala2021general}%
  \BibitemOpen
  \bibfield  {author} {\bibinfo {author} {\bibfnamefont {M.}~\bibnamefont
  {Pl{\'a}vala}},\ }\bibfield  {title} {\bibinfo {title} {General probabilistic
  theories: An introduction},\ }\href {https://arxiv.org/abs/2103.07469v2}
  {\bibfield  {journal} {\bibinfo  {journal} {arXiv preprint arXiv:2103.07469}\
  } (\bibinfo {year} {2021})}\BibitemShut {NoStop}%
\bibitem [{\citenamefont {Popescu}\ and\ \citenamefont
  {Rohrlich}(1994)}]{popescu1994quantum}%
  \BibitemOpen
  \bibfield  {author} {\bibinfo {author} {\bibfnamefont {S.}~\bibnamefont
  {Popescu}}\ and\ \bibinfo {author} {\bibfnamefont {D.}~\bibnamefont
  {Rohrlich}},\ }\bibfield  {title} {\bibinfo {title} {Quantum nonlocality as
  an axiom},\ }\href {https://doi.org/https://doi.org/10.1007/BF02058098}
  {\bibfield  {journal} {\bibinfo  {journal} {Foundations of Physics}\ }\textbf
  {\bibinfo {volume} {24}},\ \bibinfo {pages} {379} (\bibinfo {year}
  {1994})}\BibitemShut {NoStop}%
\bibitem [{\citenamefont {Paw{\l}owski}\ \emph {et~al.}(2009)\citenamefont
  {Paw{\l}owski}, \citenamefont {Paterek}, \citenamefont {Kaszlikowski},
  \citenamefont {Scarani}, \citenamefont {Winter},\ and\ \citenamefont
  {{\.Z}ukowski}}]{pawlowski2009information}%
  \BibitemOpen
  \bibfield  {author} {\bibinfo {author} {\bibfnamefont {M.}~\bibnamefont
  {Paw{\l}owski}}, \bibinfo {author} {\bibfnamefont {T.}~\bibnamefont
  {Paterek}}, \bibinfo {author} {\bibfnamefont {D.}~\bibnamefont
  {Kaszlikowski}}, \bibinfo {author} {\bibfnamefont {V.}~\bibnamefont
  {Scarani}}, \bibinfo {author} {\bibfnamefont {A.}~\bibnamefont {Winter}},\
  and\ \bibinfo {author} {\bibfnamefont {M.}~\bibnamefont {{\.Z}ukowski}},\
  }\bibfield  {title} {\bibinfo {title} {Information causality as a physical
  principle},\ }\href {https://doi.org/https://doi.org/10.1038/nature08400}
  {\bibfield  {journal} {\bibinfo  {journal} {Nature}\ }\textbf {\bibinfo
  {volume} {461}},\ \bibinfo {pages} {1101} (\bibinfo {year}
  {2009})}\BibitemShut {NoStop}%
\bibitem [{\citenamefont {Barnum}\ \emph {et~al.}(2010)\citenamefont {Barnum},
  \citenamefont {Barrett}, \citenamefont {Clark}, \citenamefont {Leifer},
  \citenamefont {Spekkens}, \citenamefont {Stepanik}, \citenamefont {Wilce},\
  and\ \citenamefont {Wilke}}]{barnum2010entropy}%
  \BibitemOpen
  \bibfield  {author} {\bibinfo {author} {\bibfnamefont {H.}~\bibnamefont
  {Barnum}}, \bibinfo {author} {\bibfnamefont {J.}~\bibnamefont {Barrett}},
  \bibinfo {author} {\bibfnamefont {L.~O.}\ \bibnamefont {Clark}}, \bibinfo
  {author} {\bibfnamefont {M.}~\bibnamefont {Leifer}}, \bibinfo {author}
  {\bibfnamefont {R.}~\bibnamefont {Spekkens}}, \bibinfo {author}
  {\bibfnamefont {N.}~\bibnamefont {Stepanik}}, \bibinfo {author}
  {\bibfnamefont {A.}~\bibnamefont {Wilce}},\ and\ \bibinfo {author}
  {\bibfnamefont {R.}~\bibnamefont {Wilke}},\ }\bibfield  {title} {\bibinfo
  {title} {Entropy and information causality in general probabilistic
  theories},\ }\href {https://doi.org/10.1088/1367-2630/12/3/033024} {\bibfield
   {journal} {\bibinfo  {journal} {New Journal of Physics}\ }\textbf {\bibinfo
  {volume} {12}},\ \bibinfo {pages} {033024} (\bibinfo {year}
  {2010})}\BibitemShut {NoStop}%
\bibitem [{\citenamefont {Masanes}\ and\ \citenamefont
  {M{\"u}ller}(2011)}]{masanes2011derivation}%
  \BibitemOpen
  \bibfield  {author} {\bibinfo {author} {\bibfnamefont {L.}~\bibnamefont
  {Masanes}}\ and\ \bibinfo {author} {\bibfnamefont {M.~P.}\ \bibnamefont
  {M{\"u}ller}},\ }\bibfield  {title} {\bibinfo {title} {A derivation of
  quantum theory from physical requirements},\ }\href
  {https://doi.org/https://doi.org/10.1088/1367-2630/13/6/063001} {\bibfield
  {journal} {\bibinfo  {journal} {New Journal of Physics}\ }\textbf {\bibinfo
  {volume} {13}},\ \bibinfo {pages} {063001} (\bibinfo {year}
  {2011})}\BibitemShut {NoStop}%
\bibitem [{\citenamefont {M\"uller}\ and\ \citenamefont
  {Ududec}(2012)}]{muller2012structure}%
  \BibitemOpen
  \bibfield  {author} {\bibinfo {author} {\bibfnamefont {M.~P.}\ \bibnamefont
  {M\"uller}}\ and\ \bibinfo {author} {\bibfnamefont {C.}~\bibnamefont
  {Ududec}},\ }\bibfield  {title} {\bibinfo {title} {Structure of reversible
  computation determines the self-duality of quantum theory},\ }\href
  {https://doi.org/https://doi.org/10.1103/PhysRevLett.108.130401} {\bibfield
  {journal} {\bibinfo  {journal} {Phys. Rev. Lett.}\ }\textbf {\bibinfo
  {volume} {108}},\ \bibinfo {pages} {130401} (\bibinfo {year}
  {2012})}\BibitemShut {NoStop}%
\bibitem [{\citenamefont {Barnum}\ \emph {et~al.}(2014)\citenamefont {Barnum},
  \citenamefont {M{\"u}ller},\ and\ \citenamefont {Ududec}}]{barnum2014higher}%
  \BibitemOpen
  \bibfield  {author} {\bibinfo {author} {\bibfnamefont {H.}~\bibnamefont
  {Barnum}}, \bibinfo {author} {\bibfnamefont {M.~P.}\ \bibnamefont
  {M{\"u}ller}},\ and\ \bibinfo {author} {\bibfnamefont {C.}~\bibnamefont
  {Ududec}},\ }\bibfield  {title} {\bibinfo {title} {Higher-order interference
  and single-system postulates characterizing quantum theory},\ }\href
  {https://doi.org/https://doi.org/10.1088/1367-2630/16/12/123029} {\bibfield
  {journal} {\bibinfo  {journal} {New Journal of Physics}\ }\textbf {\bibinfo
  {volume} {16}},\ \bibinfo {pages} {123029} (\bibinfo {year}
  {2014})}\BibitemShut {NoStop}%
\bibitem [{\citenamefont {Wakakuwa}(2021)}]{wakakuwa2021gentle}%
  \BibitemOpen
  \bibfield  {author} {\bibinfo {author} {\bibfnamefont {E.}~\bibnamefont
  {Wakakuwa}},\ }\bibfield  {title} {\bibinfo {title} {Gentle measurement as a
  principle of quantum theory},\ }\href {https://arxiv.org/abs/2103.15110}
  {\bibfield  {journal} {\bibinfo  {journal} {arXiv preprint arXiv:2103.15110}\
  } (\bibinfo {year} {2021})}\BibitemShut {NoStop}%
\bibitem [{\citenamefont {Ludwig}(1964)}]{ludwig1964}%
  \BibitemOpen
  \bibfield  {author} {\bibinfo {author} {\bibfnamefont {G.}~\bibnamefont
  {Ludwig}},\ }\bibfield  {title} {\bibinfo {title} {Versuch einer
  axiomatischen grundlegung der quantenmechanik und allgemeinerer
  physikalischer theorien},\ }\href {https://doi.org/10.1007/BF01418533}
  {\bibfield  {journal} {\bibinfo  {journal} {Zeitschrift f{\"u}r Physik}\
  }\textbf {\bibinfo {volume} {181}},\ \bibinfo {pages} {233} (\bibinfo {year}
  {1964})}\BibitemShut {NoStop}%
\bibitem [{\citenamefont {Ludwig}(1967)}]{ludwig1967}%
  \BibitemOpen
  \bibfield  {author} {\bibinfo {author} {\bibfnamefont {G.}~\bibnamefont
  {Ludwig}},\ }\bibfield  {title} {\bibinfo {title} {Attempt of an axiomatic
  foundation of quantum mechanics and more general theories, ii},\ }\href
  {https://doi.org/10.1007/BF01653647} {\bibfield  {journal} {\bibinfo
  {journal} {Communications in Mathematical Physics}\ }\textbf {\bibinfo
  {volume} {4}},\ \bibinfo {pages} {331} (\bibinfo {year} {1967})}\BibitemShut
  {NoStop}%
\bibitem [{\citenamefont {Davies}\ and\ \citenamefont
  {Lewis}(1970)}]{davies1970operational}%
  \BibitemOpen
  \bibfield  {author} {\bibinfo {author} {\bibfnamefont {E.~B.}\ \bibnamefont
  {Davies}}\ and\ \bibinfo {author} {\bibfnamefont {J.~T.}\ \bibnamefont
  {Lewis}},\ }\bibfield  {title} {\bibinfo {title} {An operational approach to
  quantum probability},\ }\href
  {https://doi.org/https://doi.org/10.1007/BF01647093} {\bibfield  {journal}
  {\bibinfo  {journal} {Communications in Mathematical Physics}\ }\textbf
  {\bibinfo {volume} {17}},\ \bibinfo {pages} {239} (\bibinfo {year}
  {1970})}\BibitemShut {NoStop}%
\bibitem [{\citenamefont {Gudder}(1973)}]{gudder1973}%
  \BibitemOpen
  \bibfield  {author} {\bibinfo {author} {\bibfnamefont {S.}~\bibnamefont
  {Gudder}},\ }\bibfield  {title} {\bibinfo {title} {{Convex structures and
  operational quantum mechanics}},\ }\href {https://doi.org/cmp/1103858551}
  {\bibfield  {journal} {\bibinfo  {journal} {Communications in Mathematical
  Physics}\ }\textbf {\bibinfo {volume} {29}},\ \bibinfo {pages} {249 }
  (\bibinfo {year} {1973})}\BibitemShut {NoStop}%
\bibitem [{\citenamefont {Ozawa}(1980)}]{ozawa1980optimal}%
  \BibitemOpen
  \bibfield  {author} {\bibinfo {author} {\bibfnamefont {M.}~\bibnamefont
  {Ozawa}},\ }\bibfield  {title} {\bibinfo {title} {Optimal measurements for
  general quantum systems},\ }\href
  {https://doi.org/https://doi.org/10.1016/0034-4877(80)90036-1} {\bibfield
  {journal} {\bibinfo  {journal} {Reports on Mathematical Physics}\ }\textbf
  {\bibinfo {volume} {18}},\ \bibinfo {pages} {11} (\bibinfo {year}
  {1980})}\BibitemShut {NoStop}%
\bibitem [{\citenamefont {Short}\ and\ \citenamefont
  {Wehner}(2010)}]{short2010entropy}%
  \BibitemOpen
  \bibfield  {author} {\bibinfo {author} {\bibfnamefont {A.~J.}\ \bibnamefont
  {Short}}\ and\ \bibinfo {author} {\bibfnamefont {S.}~\bibnamefont {Wehner}},\
  }\bibfield  {title} {\bibinfo {title} {Entropy in general physical
  theories},\ }\href
  {https://doi.org/https://doi.org/10.1088/1367-2630/12/3/033023} {\bibfield
  {journal} {\bibinfo  {journal} {New Journal of Physics}\ }\textbf {\bibinfo
  {volume} {12}},\ \bibinfo {pages} {033023} (\bibinfo {year}
  {2010})}\BibitemShut {NoStop}%
\bibitem [{\citenamefont {Kimura}\ \emph {et~al.}(2010)\citenamefont {Kimura},
  \citenamefont {Nuida},\ and\ \citenamefont
  {Imai}}]{kimura2010distinguishability}%
  \BibitemOpen
  \bibfield  {author} {\bibinfo {author} {\bibfnamefont {G.}~\bibnamefont
  {Kimura}}, \bibinfo {author} {\bibfnamefont {K.}~\bibnamefont {Nuida}},\ and\
  \bibinfo {author} {\bibfnamefont {H.}~\bibnamefont {Imai}},\ }\bibfield
  {title} {\bibinfo {title} {Distinguishability measures and entropies for
  general probabilistic theories},\ }\href
  {https://doi.org/https://doi.org/10.1016/S0034-4877(10)00025-X} {\bibfield
  {journal} {\bibinfo  {journal} {Reports on Mathematical Physics}\ }\textbf
  {\bibinfo {volume} {66}},\ \bibinfo {pages} {175} (\bibinfo {year}
  {2010})}\BibitemShut {NoStop}%
\bibitem [{\citenamefont {Kimura}\ \emph {et~al.}(2016)\citenamefont {Kimura},
  \citenamefont {Ishiguro},\ and\ \citenamefont {Fukui}}]{kimura2016entropies}%
  \BibitemOpen
  \bibfield  {author} {\bibinfo {author} {\bibfnamefont {G.}~\bibnamefont
  {Kimura}}, \bibinfo {author} {\bibfnamefont {J.}~\bibnamefont {Ishiguro}},\
  and\ \bibinfo {author} {\bibfnamefont {M.}~\bibnamefont {Fukui}},\ }\bibfield
   {title} {\bibinfo {title} {Entropies in general probabilistic theories and
  their application to the holevo bound},\ }\href
  {https://doi.org/10.1103/PhysRevA.94.042113} {\bibfield  {journal} {\bibinfo
  {journal} {Phys. Rev. A}\ }\textbf {\bibinfo {volume} {94}},\ \bibinfo
  {pages} {042113} (\bibinfo {year} {2016})}\BibitemShut {NoStop}%
\bibitem [{\citenamefont {Chiribella}\ and\ \citenamefont
  {Scandolo}(2015)}]{chiribella2015operational}%
  \BibitemOpen
  \bibfield  {author} {\bibinfo {author} {\bibfnamefont {G.}~\bibnamefont
  {Chiribella}}\ and\ \bibinfo {author} {\bibfnamefont {C.~M.}\ \bibnamefont
  {Scandolo}},\ }\bibfield  {title} {\bibinfo {title} {Operational axioms for
  diagonalizing states},\ }in\ \href
  {https://doi.org/https://doi.org/10.4204/EPTCS.195.8} {\emph {\bibinfo
  {booktitle} {{\rm Proceedings of the 12th International Workshop on} Quantum
  Physics and Logic, {\rm Oxford, U.K., July 15-17, 2015}}}},\ \bibinfo
  {series} {Electronic Proceedings in Theoretical Computer Science}, Vol.\
  \bibinfo {volume} {195},\ \bibinfo {editor} {edited by\ \bibinfo {editor}
  {\bibfnamefont {C.}~\bibnamefont {Heunen}}, \bibinfo {editor} {\bibfnamefont
  {P.}~\bibnamefont {Selinger}},\ and\ \bibinfo {editor} {\bibfnamefont
  {J.}~\bibnamefont {Vicary}}}\ (\bibinfo  {publisher} {Open Publishing
  Association},\ \bibinfo {year} {2015})\ pp.\ \bibinfo {pages}
  {96--115}\BibitemShut {NoStop}%
\bibitem [{\citenamefont {Krumm}\ \emph {et~al.}(2017)\citenamefont {Krumm},
  \citenamefont {Barnum}, \citenamefont {Barrett},\ and\ \citenamefont
  {M{\"u}ller}}]{krumm2017thermodynamics}%
  \BibitemOpen
  \bibfield  {author} {\bibinfo {author} {\bibfnamefont {M.}~\bibnamefont
  {Krumm}}, \bibinfo {author} {\bibfnamefont {H.}~\bibnamefont {Barnum}},
  \bibinfo {author} {\bibfnamefont {J.}~\bibnamefont {Barrett}},\ and\ \bibinfo
  {author} {\bibfnamefont {M.~P.}\ \bibnamefont {M{\"u}ller}},\ }\bibfield
  {title} {\bibinfo {title} {Thermodynamics and the structure of quantum
  theory},\ }\href {https://doi.org/https://doi.org/10.1088/1367-2630/aa68ef}
  {\bibfield  {journal} {\bibinfo  {journal} {New Journal of Physics}\ }\textbf
  {\bibinfo {volume} {19}},\ \bibinfo {pages} {043025} (\bibinfo {year}
  {2017})}\BibitemShut {NoStop}%
\bibitem [{\citenamefont {Barnum}\ \emph {et~al.}(2015)\citenamefont {Barnum},
  \citenamefont {Barrett}, \citenamefont {Krumm},\ and\ \citenamefont
  {M\"uller}}]{barnum2015entropy}%
  \BibitemOpen
  \bibfield  {author} {\bibinfo {author} {\bibfnamefont {H.}~\bibnamefont
  {Barnum}}, \bibinfo {author} {\bibfnamefont {J.}~\bibnamefont {Barrett}},
  \bibinfo {author} {\bibfnamefont {M.}~\bibnamefont {Krumm}},\ and\ \bibinfo
  {author} {\bibfnamefont {M.~P.}\ \bibnamefont {M\"uller}},\ }\bibfield
  {title} {\bibinfo {title} {Entropy, majorization and thermodynamics in
  general probabilistic theories},\ }in\ \href
  {https://doi.org/https://doi.org/10.4204/EPTCS.195.4} {\emph {\bibinfo
  {booktitle} {{\rm Proceedings of the 12th International Workshop on} Quantum
  Physics and Logic, {\rm Oxford, U.K., July 15-17, 2015}}}},\ \bibinfo
  {series} {Electronic Proceedings in Theoretical Computer Science}, Vol.\
  \bibinfo {volume} {195},\ \bibinfo {editor} {edited by\ \bibinfo {editor}
  {\bibfnamefont {C.}~\bibnamefont {Heunen}}, \bibinfo {editor} {\bibfnamefont
  {P.}~\bibnamefont {Selinger}},\ and\ \bibinfo {editor} {\bibfnamefont
  {J.}~\bibnamefont {Vicary}}}\ (\bibinfo  {publisher} {Open Publishing
  Association},\ \bibinfo {year} {2015})\ pp.\ \bibinfo {pages}
  {43--58}\BibitemShut {NoStop}%
\bibitem [{\citenamefont {Krumm}(2015)}]{krumm2015thermodynamics}%
  \BibitemOpen
  \bibfield  {author} {\bibinfo {author} {\bibfnamefont {M.}~\bibnamefont
  {Krumm}},\ }\bibfield  {title} {\bibinfo {title} {Thermodynamics and the
  structure of quantum theory as a generalized probabilistic theory,
  {H}eidelberg {U}niversity, {M}aster thesis},\ }\href
  {https://arxiv.org/abs/1508.03299} {\bibfield  {journal} {\bibinfo  {journal}
  {arXiv preprint arXiv:1508.03299}\ } (\bibinfo {year} {2015})}\BibitemShut
  {NoStop}%
\bibitem [{\citenamefont {Groenewold}(1971)}]{groenewold1971problem}%
  \BibitemOpen
  \bibfield  {author} {\bibinfo {author} {\bibfnamefont {H.~J.}\ \bibnamefont
  {Groenewold}},\ }\bibfield  {title} {\bibinfo {title} {A problem of
  information gain by quantal measurements},\ }\href
  {https://doi.org/https://doi.org/10.1007/BF00815357} {\bibfield  {journal}
  {\bibinfo  {journal} {International Journal of Theoretical Physics}\ }\textbf
  {\bibinfo {volume} {4}},\ \bibinfo {pages} {327} (\bibinfo {year}
  {1971})}\BibitemShut {NoStop}%
\bibitem [{\citenamefont {Ozawa}(1986)}]{ozawa1986information}%
  \BibitemOpen
  \bibfield  {author} {\bibinfo {author} {\bibfnamefont {M.}~\bibnamefont
  {Ozawa}},\ }\bibfield  {title} {\bibinfo {title} {On information gain by
  quantum measurements of continuous observables},\ }\href
  {https://doi.org/https://doi.org/10.1063/1.527179} {\bibfield  {journal}
  {\bibinfo  {journal} {Journal of mathematical physics}\ }\textbf {\bibinfo
  {volume} {27}},\ \bibinfo {pages} {759} (\bibinfo {year} {1986})}\BibitemShut
  {NoStop}%
\bibitem [{\citenamefont {Ozawa}()}]{no-signaling-future}%
  \BibitemOpen
  \bibfield  {author} {\bibinfo {author} {\bibfnamefont {M.}~\bibnamefont
  {Ozawa}},\ }\href@noop {} {}\bibinfo {howpublished} {private communication,
  see~\cite{dariano_2010}}\BibitemShut {NoStop}%
\bibitem [{\citenamefont {D'Ariano}(2010)}]{dariano_2010}%
  \BibitemOpen
  \bibfield  {author} {\bibinfo {author} {\bibfnamefont {G.~M.}\ \bibnamefont
  {D'Ariano}},\ }\bibfield  {title} {\bibinfo {title} {Probabilistic theories:
  What is special about quantum mechanics?},\ }in\ \href@noop {} {\emph
  {\bibinfo {booktitle} {{P}hilosophy of {Q}uantum {I}nformation and
  {E}ntanglement}}},\ \bibinfo {editor} {edited by\ \bibinfo {editor}
  {\bibfnamefont {A.}~\bibnamefont {Bokulich}}\ and\ \bibinfo {editor}
  {\bibfnamefont {G.}~\bibnamefont {{J}aeger}}}\ (\bibinfo  {publisher}
  {Cambridge University Press, Cambridge},\ \bibinfo {year} {2010})\ pp.\
  \bibinfo {pages} {85--126}\BibitemShut {NoStop}%
\bibitem [{Note1()}]{Note1}%
  \BibitemOpen
  \bibinfo {note} {Notice that we are not assuming that \protect \textit {any}
  family of affine maps constitutes a legitimate test of the theory: further
  conditions, most notably that of complete positivity, may restrict the set of
  possible tests. However, this point is irrelevant for the present
  discussion.}\BibitemShut {Stop}%
\bibitem [{\citenamefont {Ozawa}(1984)}]{ozawa1984quantum}%
  \BibitemOpen
  \bibfield  {author} {\bibinfo {author} {\bibfnamefont {M.}~\bibnamefont
  {Ozawa}},\ }\bibfield  {title} {\bibinfo {title} {Quantum measuring processes
  of continuous observables},\ }\href
  {https://doi.org/https://doi.org/10.1063/1.526000} {\bibfield  {journal}
  {\bibinfo  {journal} {Journal of Mathematical Physics}\ }\textbf {\bibinfo
  {volume} {25}},\ \bibinfo {pages} {79} (\bibinfo {year} {1984})}\BibitemShut
  {NoStop}%
\bibitem [{\citenamefont {Buscemi}\ \emph {et~al.}(2004)\citenamefont
  {Buscemi}, \citenamefont {D'Ariano},\ and\ \citenamefont
  {Perinotti}}]{buscemi2004repeatable}%
  \BibitemOpen
  \bibfield  {author} {\bibinfo {author} {\bibfnamefont {F.}~\bibnamefont
  {Buscemi}}, \bibinfo {author} {\bibfnamefont {G.~M.}\ \bibnamefont
  {D'Ariano}},\ and\ \bibinfo {author} {\bibfnamefont {P.}~\bibnamefont
  {Perinotti}},\ }\bibfield  {title} {\bibinfo {title} {There exist
  nonorthogonal quantum measurements that are perfectly repeatable},\ }\href
  {https://doi.org/10.1103/PhysRevLett.92.070403} {\bibfield  {journal}
  {\bibinfo  {journal} {Phys. Rev. Lett.}\ }\textbf {\bibinfo {volume} {92}},\
  \bibinfo {pages} {070403} (\bibinfo {year} {2004})}\BibitemShut {NoStop}%
\bibitem [{\citenamefont {Buscemi}\ \emph {et~al.}(2005)\citenamefont
  {Buscemi}, \citenamefont {Keyl}, \citenamefont {D’Ariano}, \citenamefont
  {Perinotti},\ and\ \citenamefont {Werner}}]{buscemi2005clean}%
  \BibitemOpen
  \bibfield  {author} {\bibinfo {author} {\bibfnamefont {F.}~\bibnamefont
  {Buscemi}}, \bibinfo {author} {\bibfnamefont {M.}~\bibnamefont {Keyl}},
  \bibinfo {author} {\bibfnamefont {G.~M.}\ \bibnamefont {D’Ariano}},
  \bibinfo {author} {\bibfnamefont {P.}~\bibnamefont {Perinotti}},\ and\
  \bibinfo {author} {\bibfnamefont {R.~F.}\ \bibnamefont {Werner}},\ }\bibfield
   {title} {\bibinfo {title} {Clean positive operator valued measures},\ }\href
  {https://doi.org/https://doi.org/10.1063/1.2008996} {\bibfield  {journal}
  {\bibinfo  {journal} {Journal of Mathematical Physics}\ }\textbf {\bibinfo
  {volume} {46}},\ \bibinfo {pages} {082109} (\bibinfo {year}
  {2005})}\BibitemShut {NoStop}%
\bibitem [{\citenamefont {H{\"a}nggi}\ and\ \citenamefont
  {Wehner}(2013)}]{hanggi2013violation}%
  \BibitemOpen
  \bibfield  {author} {\bibinfo {author} {\bibfnamefont {E.}~\bibnamefont
  {H{\"a}nggi}}\ and\ \bibinfo {author} {\bibfnamefont {S.}~\bibnamefont
  {Wehner}},\ }\bibfield  {title} {\bibinfo {title} {A violation of the
  uncertainty principle implies a violation of the second law of
  thermodynamics},\ }\href {https://doi.org/https://doi.org/10.1038/ncomms2665}
  {\bibfield  {journal} {\bibinfo  {journal} {Nature communications}\ }\textbf
  {\bibinfo {volume} {4}},\ \bibinfo {pages} {1} (\bibinfo {year} {2013})},\
  \bibinfo {note} {arXiv:1205.6894}\BibitemShut {NoStop}%
\bibitem [{\citenamefont {Barnum}\ and\ \citenamefont
  {Hilgert}(2020)}]{barnum2020spectral}%
  \BibitemOpen
  \bibfield  {author} {\bibinfo {author} {\bibfnamefont {H.}~\bibnamefont
  {Barnum}}\ and\ \bibinfo {author} {\bibfnamefont {J.}~\bibnamefont
  {Hilgert}},\ }\bibfield  {title} {\bibinfo {title} {Spectral properties of
  convex bodies},\ }\href
  {https://www.heldermann.de/JLT/JLT30/JLT302/jlt30019.htm} {\bibfield
  {journal} {\bibinfo  {journal} {J. Lie Theory}\ }\textbf {\bibinfo {volume}
  {30}},\ \bibinfo {pages} {315} (\bibinfo {year} {2020})}\BibitemShut
  {NoStop}%
\bibitem [{\citenamefont {Danageozian}\ \emph {et~al.}(2022)\citenamefont
  {Danageozian}, \citenamefont {Wilde},\ and\ \citenamefont
  {Buscemi}}]{danageozian2021thermodynamic}%
  \BibitemOpen
  \bibfield  {author} {\bibinfo {author} {\bibfnamefont {A.}~\bibnamefont
  {Danageozian}}, \bibinfo {author} {\bibfnamefont {M.~M.}\ \bibnamefont
  {Wilde}},\ and\ \bibinfo {author} {\bibfnamefont {F.}~\bibnamefont
  {Buscemi}},\ }\bibfield  {title} {\bibinfo {title} {Thermodynamic constraints
  on quantum information gain and error correction: A triple trade-off},\
  }\href {https://doi.org/10.1103/PRXQuantum.3.020318} {\bibfield  {journal}
  {\bibinfo  {journal} {PRX Quantum}\ }\textbf {\bibinfo {volume} {3}},\
  \bibinfo {pages} {020318} (\bibinfo {year} {2022})}\BibitemShut {NoStop}%
\bibitem [{\citenamefont {Watanabe}(1955)}]{watanabe55}%
  \BibitemOpen
  \bibfield  {author} {\bibinfo {author} {\bibfnamefont {S.}~\bibnamefont
  {Watanabe}},\ }\bibfield  {title} {\bibinfo {title} {Symmetry of physical
  laws. part iii. prediction and retrodiction},\ }\href
  {https://doi.org/10.1103/RevModPhys.27.179} {\bibfield  {journal} {\bibinfo
  {journal} {Rev. Mod. Phys.}\ }\textbf {\bibinfo {volume} {27}},\ \bibinfo
  {pages} {179} (\bibinfo {year} {1955})}\BibitemShut {NoStop}%
\bibitem [{\citenamefont {Watanabe}(1965)}]{watanabe65}%
  \BibitemOpen
  \bibfield  {author} {\bibinfo {author} {\bibfnamefont {S.}~\bibnamefont
  {Watanabe}},\ }\bibfield  {title} {\bibinfo {title} {Conditional
  probabilities in physics},\ }\href
  {https://doi.org/https://doi.org/10.1143/PTPS.E65.135} {\bibfield  {journal}
  {\bibinfo  {journal} {Progr. Theor. Phys. Suppl.}\ }\textbf {\bibinfo
  {volume} {E65}},\ \bibinfo {pages} {135} (\bibinfo {year}
  {1965})}\BibitemShut {NoStop}%
\bibitem [{\citenamefont {Buscemi}\ and\ \citenamefont
  {Scarani}(2021)}]{buscemi2021fluctuation}%
  \BibitemOpen
  \bibfield  {author} {\bibinfo {author} {\bibfnamefont {F.}~\bibnamefont
  {Buscemi}}\ and\ \bibinfo {author} {\bibfnamefont {V.}~\bibnamefont
  {Scarani}},\ }\bibfield  {title} {\bibinfo {title} {Fluctuation theorems from
  bayesian retrodiction},\ }\href
  {https://doi.org/https://doi.org/10.1103/PhysRevE.103.052111} {\bibfield
  {journal} {\bibinfo  {journal} {Phys. Rev. E}\ }\textbf {\bibinfo {volume}
  {103}},\ \bibinfo {pages} {052111} (\bibinfo {year} {2021})}\BibitemShut
  {NoStop}%
\bibitem [{\citenamefont {Aw}\ \emph {et~al.}(2021)\citenamefont {Aw},
  \citenamefont {Buscemi},\ and\ \citenamefont {Scarani}}]{aw-buscemi-scarani}%
  \BibitemOpen
  \bibfield  {author} {\bibinfo {author} {\bibfnamefont {C.~C.}\ \bibnamefont
  {Aw}}, \bibinfo {author} {\bibfnamefont {F.}~\bibnamefont {Buscemi}},\ and\
  \bibinfo {author} {\bibfnamefont {V.}~\bibnamefont {Scarani}},\ }\bibfield
  {title} {\bibinfo {title} {Fluctuation theorems with retrodiction rather than
  reverse processes},\ }\href
  {https://doi.org/https://doi.org/10.1116/5.0060893} {\bibfield  {journal}
  {\bibinfo  {journal} {AVS Quantum Science}\ }\textbf {\bibinfo {volume}
  {3}},\ \bibinfo {pages} {045601} (\bibinfo {year} {2021})}\BibitemShut
  {NoStop}%
\bibitem [{\citenamefont {Perinotti}\ \emph {et~al.}(2022)\citenamefont
  {Perinotti}, \citenamefont {Tosini},\ and\ \citenamefont
  {Vaglini}}]{perinotti2021shannon}%
  \BibitemOpen
  \bibfield  {author} {\bibinfo {author} {\bibfnamefont {P.}~\bibnamefont
  {Perinotti}}, \bibinfo {author} {\bibfnamefont {A.}~\bibnamefont {Tosini}},\
  and\ \bibinfo {author} {\bibfnamefont {L.}~\bibnamefont {Vaglini}},\
  }\bibfield  {title} {\bibinfo {title} {Shannon theory beyond quantum:
  Information content of a source},\ }\href
  {https://doi.org/10.1103/PhysRevA.105.052222} {\bibfield  {journal} {\bibinfo
   {journal} {Phys. Rev. A}\ }\textbf {\bibinfo {volume} {105}},\ \bibinfo
  {pages} {052222} (\bibinfo {year} {2022})}\BibitemShut {NoStop}%
\bibitem [{\citenamefont {Hayashi}\ and\ \citenamefont
  {Tajima}(2017)}]{hayashi2017measurement}%
  \BibitemOpen
  \bibfield  {author} {\bibinfo {author} {\bibfnamefont {M.}~\bibnamefont
  {Hayashi}}\ and\ \bibinfo {author} {\bibfnamefont {H.}~\bibnamefont
  {Tajima}},\ }\bibfield  {title} {\bibinfo {title} {Measurement-based
  formulation of quantum heat engines},\ }\href
  {https://doi.org/https://doi.org/10.1103/PhysRevA.95.032132} {\bibfield
  {journal} {\bibinfo  {journal} {Phys. Rev. A}\ }\textbf {\bibinfo {volume}
  {95}},\ \bibinfo {pages} {032132} (\bibinfo {year} {2017})}\BibitemShut
  {NoStop}%
\bibitem [{\citenamefont {Arai}\ \emph {et~al.}(2019)\citenamefont {Arai},
  \citenamefont {Yoshida},\ and\ \citenamefont {Hayashi}}]{arai2019perfect}%
  \BibitemOpen
  \bibfield  {author} {\bibinfo {author} {\bibfnamefont {H.}~\bibnamefont
  {Arai}}, \bibinfo {author} {\bibfnamefont {Y.}~\bibnamefont {Yoshida}},\ and\
  \bibinfo {author} {\bibfnamefont {M.}~\bibnamefont {Hayashi}},\ }\bibfield
  {title} {\bibinfo {title} {Perfect discrimination of non-orthogonal separable
  pure states on bipartite system in general probabilistic theory},\ }\href
  {https://doi.org/https://doi.org/10.1088/1751-8121/ab4b64} {\bibfield
  {journal} {\bibinfo  {journal} {Journal of Physics A: Mathematical and
  Theoretical}\ }\textbf {\bibinfo {volume} {52}},\ \bibinfo {pages} {465304}
  (\bibinfo {year} {2019})}\BibitemShut {NoStop}%
\bibitem [{\citenamefont {Friedland}\ \emph {et~al.}(2011)\citenamefont
  {Friedland}, \citenamefont {Li}, \citenamefont {Poon},\ and\ \citenamefont
  {Sze}}]{friedland2011automorphism}%
  \BibitemOpen
  \bibfield  {author} {\bibinfo {author} {\bibfnamefont {S.}~\bibnamefont
  {Friedland}}, \bibinfo {author} {\bibfnamefont {C.-K.}\ \bibnamefont {Li}},
  \bibinfo {author} {\bibfnamefont {Y.-T.}\ \bibnamefont {Poon}},\ and\
  \bibinfo {author} {\bibfnamefont {N.-S.}\ \bibnamefont {Sze}},\ }\bibfield
  {title} {\bibinfo {title} {The automorphism group of separable states in
  quantum information theory},\ }\href
  {https://doi.org/https://doi.org/10.1063/1.3578015} {\bibfield  {journal}
  {\bibinfo  {journal} {Journal of Mathematical Physics}\ }\textbf {\bibinfo
  {volume} {52}},\ \bibinfo {pages} {042203} (\bibinfo {year}
  {2011})}\BibitemShut {NoStop}%
\end{thebibliography}%

\end{document}